\documentclass[runningheads]{llncs}
\usepackage{graphicx}
\usepackage[sort]{cite}
\usepackage{enumerate}
\usepackage{booktabs}
\usepackage{tabu}
\usepackage[ruled,lined,boxed,linesnumbered]{algorithm2e}
\usepackage{amsmath,amssymb}
\usepackage{todonotes}
\usepackage{algorithmic}

\newtheorem{Reduction Rule}{Reduction Rule}

\usepackage{thmtools}
\usepackage{thm-restate}

\usepackage{hyperref}

\usepackage{cleveref}

\usepackage{microtype}%if unwanted, comment out or use option "draft"
\usepackage{xspace}
\usepackage{color}

\newcommand{\stpath}{{$s$-$t$ path}\xspace}

\newcommand{\lstpair}{{local $s$-$t$ pair}\xspace}
\newcommand{\lstpairs}{{local $s$-$t$ pairs}\xspace}

\newcommand{\stpaths}{{$s$-$t$ paths}\xspace}

\newcommand{\dcm}{{dual connected modulator}\xspace}

\newcommand{\tp}{{\sc Tracking Paths}\xspace}

\newcommand{\tsp}{{\sc Tracking Shortest Paths}\xspace}

\newcommand{\Oh}{\mathcal{O}}

\newcommand{\defproblem}[3]{
  \vspace{1mm}
\noindent\fbox{
  \begin{minipage}{0.96\textwidth}
  \begin{tabular*}{\textwidth}{@{\extracolsep{\fill}}lr} #1 \\ \end{tabular*}
  {\bf{Input:}} #2  \\
  {\bf{Question:}} #3
  \end{minipage}
  }
  \vspace{1mm}
}

\begin{document}

\title{Structural Parameterizations of Tracking Paths Problem}

\author{Pratibha Choudhary \inst{1} \and Venkatesh Raman \inst{2}}

\institute{Indian Institute of Technology Jodhpur, Jodhpur, India.\\
\email{pratibhac247@gmail.com}
\and
Institute of Mathematical Sciences, HBNI, Chennai, India.\\
\email{vraman@imsc.res.in}
}

\authorrunning{P. Choudhary and V. Raman}
\maketitle              % typeset the header of the contribution
\begin{abstract}
Given a graph $G$ with source and destination vertices $s,t\in V(G)$ respectively, \textsc{Tracking Paths} asks for a minimum set of vertices $T\subseteq V(G)$, such that the sequence of vertices encountered in each simple path from $s$ to $t$ is unique. The problem was proven \textsc{NP}-hard~\cite{tr-j} and was found to admit a quadratic kernel when parameterized by the size of the desired solution~\cite{quadratic}. Following recent trends, for the first time, we study \textsc{Tracking Paths} with respect to structural parameters of the input graph, parameters that measure how far the input graph is, from an easy instance. We prove that \textsc{Tracking Paths} admits fixed-parameter tractable (\textsc{FPT}) algorithms when parameterized by the size of vertex cover, and the size of cluster vertex deletion set for the input graph.

%\and matching}
\end{abstract}

\section{Introduction}
Graph theory plays a fundamental role in modeling many real world problems related to (but not limited to) road networks, traffic monitoring, world wide web, social networks and circuit design. One of the graph theoretic problems studied in recent years is \tp: Given a graph, find a set of vertices that can help uniquely distinguish all simple paths between a given source and destination in the input graph. The problem finds applications in secure facility object tracking, tracing data packets in network, identifying source of fake news on social media, and tracking objects in wireless sensor networks. 

More formally, let $\overrightarrow{V}(P)$ be the sequence of vertices in a path $P$. A \textit{tracking set} for a graph $G$ with source $s$ and destination $t$ is a subset $T$ of vertices such that for any two distinct \stpaths $P_1$ and $P_2$, $\overrightarrow{V}(P_1) \neq \overrightarrow{V}(P_2)$, and the \tp problem is defined as follows:

\defproblem{\tp $(G,s,t)$}{An undirected  graph $G=(V,E)$ with terminal vertices $s$ and $t$.}
{Find a minimum cardinality tracking set $T$ for $G$.}
\medskip

The problem was first studied by Banik et al.~\cite{ciac17}, where the problem was restricted to distinguishing all {\it shortest} \stpaths in a graph. The authors proved the problem \textsc{NP}-hard and \textsc{APX}-hard, and gave a $2$-approximation algorithm for \tsp in planar graphs. 

Parameterization of a problem involves associating the problem with an integer $k$. A parameterized problem is said to admit a \textit{fixed-parameter tractable}(\textsc{FPT}) algorithm if there exists an algorithm with running time of the type $f(k).n^{\mathcal{O}(1)}$, where $f$ is a computable function, $k$ is the parameter and $n$ is the input size.
\tsp was proven to be \textsc{FPT} when parameterized by the size of tracking set~\cite{caldam18}. Bilò et al.~\cite{guido-cubic} gave an FPT algorithm for the case when there are multiple sources and destinations, the parameter being the maximum number of vertices equidistant from the source (or destination).

\tp (not just tracking shortest paths) was proven to be \textsc{NP}-hard for general graphs~\cite{tr-j}. Note that from the definition, it is not even clear how to verify if a subset of vertices forms a tracking set, as there can be exponentially many $s$-$t$ paths. Through an equivalent characterization, a polynomial time algorithm was shown~\cite{tr-j} for this task, thus proving the problem \textsc{NP}.

\begin{theorem}~\cite{tr-j}
\label{theorem:np}
\tp belongs to \textsc{NP}, i.e. for a graph $G$ and a set of vertices $T\subseteq V(G)$, there exists a polynomial time algorithm to verify if $T$ is a tracking set for $G$.
\end{theorem}

The problem was shown to be \textsc{FPT} when parameterized by the size of tracking set, by showing the existence of a polynomial kernel~\cite{tr-j,quadratic}. A kernel for a parameterized problem is an equivalent instance of the given problem, whose size (of the reduced new instance) is bounded by a function of just the parameter. Kernelization (the process of deriving a kernel) is usually achieved through \textit{Reduction Rules} which are preprocessing operations. A reduction rule is said to be safe if the new instance is equivalent to the original one, i.e. the original instance is a YES instance if and only if the new one is a YES instance. %The kernel for \tp was later improved by showing existence of a quadratic kernel~\cite{quadratic}. %It was also proven that it is \textsc{W}[1]-hard to find whether there exists a tracking set of size $n-k$, where $n$ is the number of vertices in the graph and $k$ is an integer. 
Generalized combinatorial versions of \tsp have been studied in~\cite{caldam18} and~\cite{tr1-j}. Eppstein et al. studied \tp for planar graphs~\cite{ep-planar}. They showed the problem \textsc{NP}-hard and gave a $4$-approximation algorithm. They also gave a linear time algorithm for bounded clique-width graphs. Recently we gave polynomial time algorithms for some restricted cases of \tp~\cite{iwoca}.

%It is common for real-world problem instances to have some kind of structure, and sometimes it can be easy to capture this structure in form of a small valued parameter.
For a parameterized problem, although output size is a natural parameter, recent years have seen increasing attention on parameters related to structure of the input~\cite{fundamentals-pc,para-ecology,graph-coloring,dipt-fvs}.
So far, parameterized analysis of \tp has been done only with respect to the output size. In this paper, we study \tp parameterized by the size of vertex cover and the size of cluster vertex deletion set. For a graph $G=(V,E)$, a \textit{vertex cover} is a set of vertices that covers all edges, i.e. the union of these vertices includes at least one endpoint of each edge in $E$. Removal of a vertex cover leaves the graph edgeless. For $G$, \textit{cluster vertex deletion set} is the set of vertices whose removal converts $G$ into a cluster graph: a graph whose each component is a clique.

Edgeless graphs do not need any trackers (since they lack \stpaths) and a tracking set can be found in polynomial time for cluster graphs (proven later in the paper). Hence, it is an interesting question to analyze whether there exists an \textsc{FPT} algorithm to solve \tp for graphs that are $k$ vertices away from an edgeless graph or a cluster graph. Usually, the quest is to look for the smallest possible parameter for which the problem at hand is fixed-parameter tractable. In general, the size of a vertex cover can be both larger or smaller than the size a tracking set for a graph. A graph with long paths of degree two vertices can have a vertex cover larger than the size of a tracking set.
%In general, the size of a vertex cover can be arbitrarily large compared to the size of a tracking set in an $s$-$t$ graph. However, for a graph preprocessed under the known set of reduction rules, the size of a minimum tracking set can be larger than that of a minimum vertex cover. 
While a denser graph can have a tracking set bigger than the size of a vertex cover. See Figure~\ref{fig:vc-ts}. Here the circled vertices represent a vertex cover. However, all vertices except $s,t$ need to be part of a tracking set.

\begin{figure}[ht]
\centering
\includegraphics[scale=0.35]{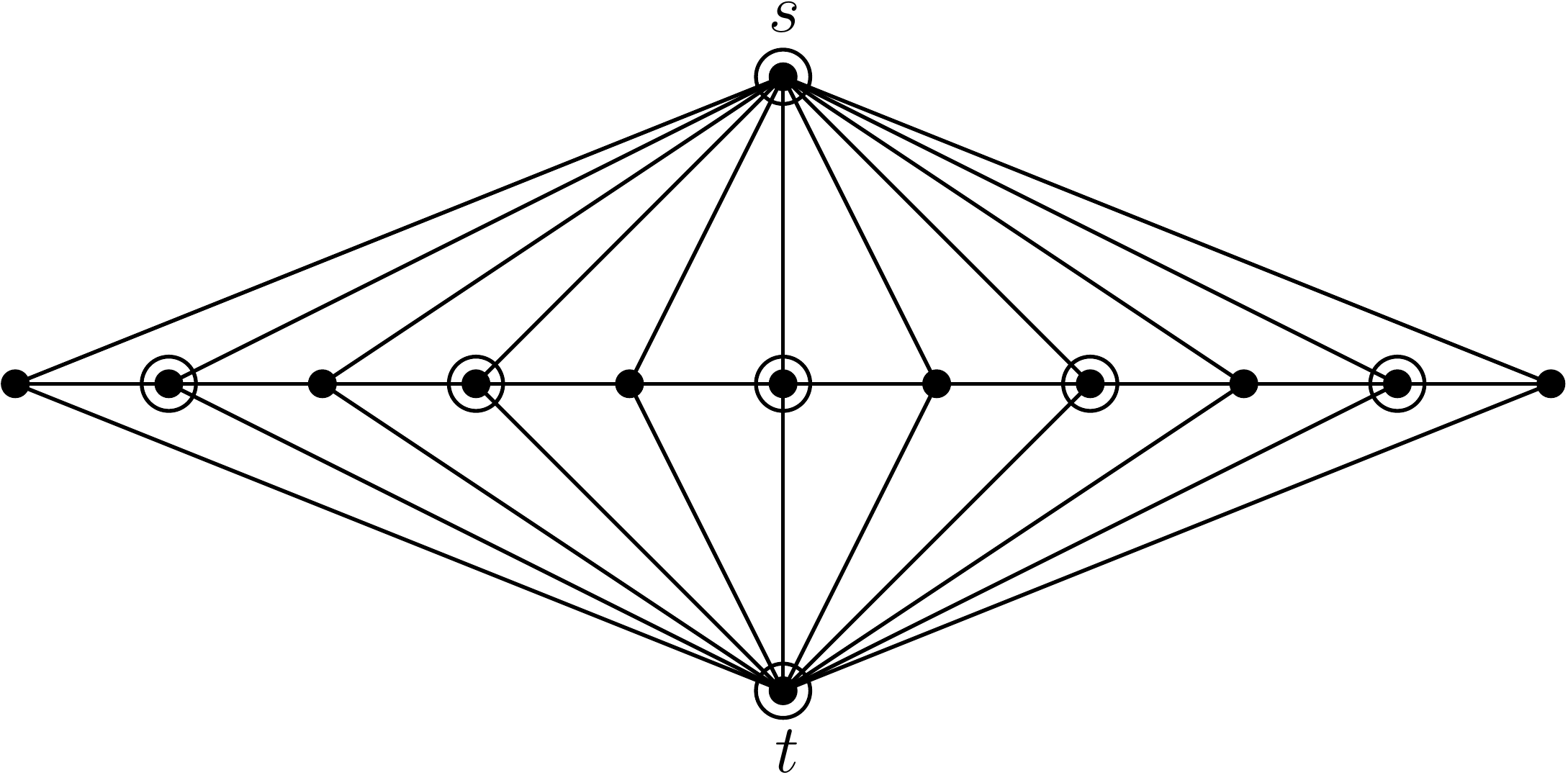} 
\caption{Graph with tracking set larger than a vertex cover (circled vertices)} 
\label{fig:vc-ts}
\end{figure}

\textbf{Our Approach.}
The usual challenge with structural parameterization is that the parameter does not drop by rules that utilize properties of the output.
We start by applying some known preprocessing rules and then use some structural properties to mark vertices as trackers that definitely need to belong to any tracking set. 
Then we bound the number of vertices that are left unmarked as a function of the parameter. Finally we try all subsets of the unmarked vertices to find which among them can be trackers.

To design FPT algorithms for the two parameterizations we consider in this paper, we first define an intermediate parameter, and that is the size of what we call a \textit{Dual Connected Modulator}. For a graph $G=(V,E)$, a set of vertices $S\subseteq V$ is a \textbf{dual connected modulator} (DCM) if every vertex in $V\setminus S$ has at least two neighbours in $S$ and has an additional property $\Pi$. For parameterization by the size of vertex cover or cluster vertex deletion set, it suffices that $\Pi$ is a disjoint union of cliques. %However, we will show later that some more properties can also be used for $\Pi$.

We will first give an FPT algorithm for \tp parameterized by the size of a \dcm. Then we show how this algorithm can be used to give FPT algorithms for \tp parameterized by the size of vertex cover and cluster vertex deletion set.%Then we show how this algorithm can be used to give FPT algorithms for \tp parameterized by the size of vertex cover, maximum matching and cluster vertex deletion set.

\section{Notations and Definitions}
Throughout the paper, we assume graphs to be simple (no self loops or multi-edges).
We assume that the input graph contains a unique source $s$ and a unique destination $t$ ($s$ and $t$ are known), and we aim to find a tracking set that can distinguish all simple paths between $s$ and $t$. Here $s$ and $t$ are also referred as the terminal vertices. If $a,b\in V$, then unless otherwise stated, $\{a,b\}$ represents the set of vertices $a,b$, and $(a,b)$ represents an edge between $a$ and $b$. For a vertex $v\in V$, the \textit{neighbourhood} of $v$ is denoted by $N(v)=\{x \mid (x,v)\in E\}$. \textit{Degree} of a vertex $v$ is denoted by $deg(v)=|N(v)|$. For set of vertices $V'$, $G(V')$ denotes the graph induced by vertices in $V'$.
For a subgraph $G'$ of $G$, $V(G')$ represents the vertex set of $G'$ and $E(G')$ represents those edges whose both endpoints belong to $V(G')$. We use $G'\subseteq G$ to denote that $G'$ is a subgraph of $G$. For a vertex $v\in V$ and a subgraph $G'$, $N_{G'}(v)=N(v)\cap V(G')$ and $deg_{G'}(v)=|N(v)\cap V(G')|$. For a subset of vertices $V'\subseteq V$ we use $N(V')$ to denote $\bigcup_{v\in V'} N(v)$. For a graph $G$ and a set of vertices $S\subseteq V(G)$, $G-S$ denotes the subgraph induced by the vertex set $V(G)\setminus V(S)$. For $A,B\subseteq V(G)$, $A\uplus B$ denotes that $A$ and $B$ are vertex disjoint partitions of graph $G$.
Let $P_1$ be a path between vertices $a$ and $b$, and $P_2$ be a path between vertices $b$ and $c$, such that $V(P_1)\cap V(P_2)=\{b\}$. By $P_1 . P_2$, we denote the path between $a$ and $c$, formed by concatenating paths $P_1$ and $P_2$ at $b$. Two paths $P_1$ and $P_2$ are said to be \textit{vertex disjoint} if their vertex sets do not intersect except possibly at the end points, i.e. $V(P_1)\cap V(P_2) \subseteq \{a,b\}$, where $a$ and $b$ are the starting and end points of the paths.
For details on parameterized complexity please refer to~\cite{book,fundamentals-pc,flum-grohe}.

\section{Parameterization by Dual Connected Modulator}
\label{sec:dcm}

In this section, we give an FPT algorithm for \tp parameterized by the size of a \textit{Dual Connected Modulator}. Recall that for a graph $G$, a subset of vertices $S\subseteq V(G)$ is a DCM if every vertex in $V(G)\setminus S$ has at least two neighbours in $S$, and has an additional property $\Pi$.

\defproblem{\tp/DCM $(G,s,t,S,k)$}{An undirected  graph $G=(V,E)$ with terminal vertices $s$ and $t$, and a dual connected modulator $S\subseteq V(G)$ for $G$, such that $|S|=k$.}
{Find a minimum cardinality tracking set  $T$ for $G$.}
\medskip

The main idea of the algorithm is to first guess how $S$ intersects with a tracking set $T$ in $G$, and then for each such guess, analyze the graph structures across the partition $S\uplus (G-S)$ and mark as many vertices as possible, as trackers. In the process, we give an upper bound for the number of vertices left unmarked in $G-S$. Finally, we consider all possible subsets of unmarked vertices in $G-S$ as trackers, and together with the set of already marked vertices, we verify if they form a tracking set for the graph $G$ using Theorem~\ref{theorem:np}. Initially none of the vertices in $V$ are marked as trackers. We start by recalling some preprocessing rules from previous work. % ~\cite{tr-j}, \cite{quadratic} and \cite{ep-planar}. 

\begin{Reduction Rule}\cite{tr-j}
\label{red:stpath}
If there exists a vertex or an edge that does not participate in any \stpath then delete it.
\end{Reduction Rule}

\begin{Reduction Rule}\cite{quadratic} (Rephrased)
\label{red:no-deg-one}
If $V\setminus\{s,t\}=\emptyset$, then return an $\emptyset$ as a solution. Else, if degree of $s$ (or $t$) is $1$ and $N(s)\neq t$ ($N(t)\neq s$), then delete $s$($t$), and label the vertex adjacent to it  as $s$($t$). 
\end{Reduction Rule}

\begin{Reduction Rule}\cite{ep-planar}
\label{red:degree-two}
Let $u,v\in V(G)$ such that $deg(u)=deg(v)=2$, $N(v)=\{u,w\}$, then delete $v$ and introduce an edge between $u$ and $w$.
\end{Reduction Rule}

We apply above rules repeatedly as long as they are applicable. We refer to the graph obtained after applying the above rules as a {\it reduced graph}.
Note that in the reduced graph, there is no vertex with degree less than or equal to $1$, each vertex and edge participates in an \stpath and there are no long degree $2$ paths (paths with consequent degree $2$ vertices).
Throughout the paper, after application of each reduction rule, we retain the notations of $G,S$ and $k$ to refer to the graph, modulator, and size of the modulator. For now, we assume that the application of above reduction rules does not destroy any properties of the modulator. Later, while analyzing specific graph parameters, we shall tweak the rules in order to maintain the modulator properties.

\begin{figure}[ht]
\centering
\includegraphics[scale=0.3]{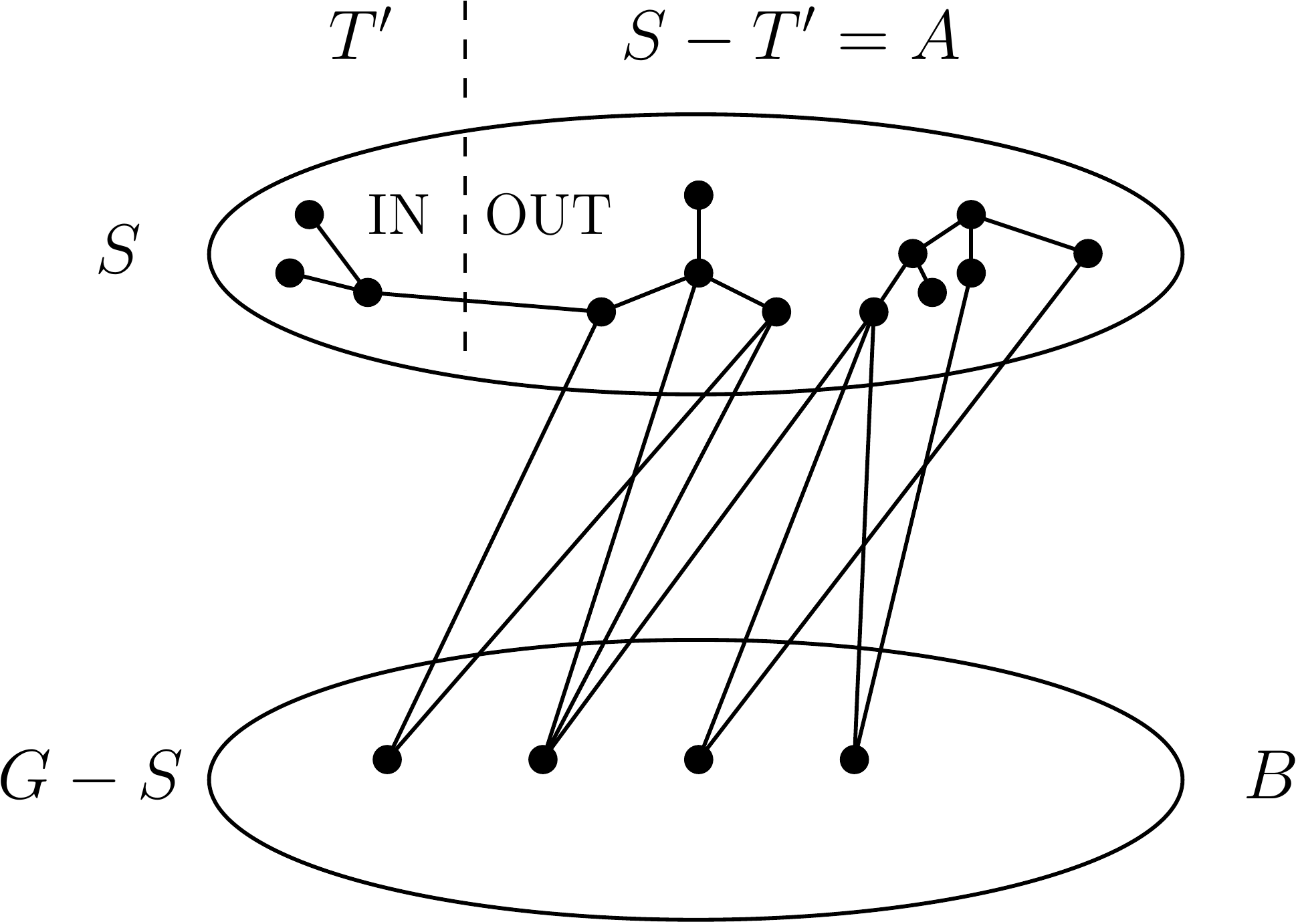} 
\caption{Disjoint tracking set problem} 
\label{fig:disjoint-vc}
\end{figure}

Let $T$ be a minimum tracking set that will be output by the algorithm, and $T'= T\cap S$ be the subset of $T$ that belongs to $S$. We attempt to guess the vertices that belong to $T'$, and towards this we simply consider all possible subsets of $S$. Hence, there are $2^k$ possible choices for $T'$. For each guess $T'$, observe that the vertices in $S\setminus T'$ cannot belong to $T$. Thus we need to find a tracking set for $G$, that is disjoint from $S\setminus T'$. We refer to this problem as \textit{disjoint tracking set}. See Figure~\ref{fig:disjoint-vc}. $S$ represents a DCM for $G$ that is received as a part of the input and $G-S$ represents the graph formed after removal of $S$ from $G$. Recall that each vertex in $G-S$ has two neighbours in $S$.
First we rule out those guesses for $T'$ that can be easily discarded. 
It is known from~\cite{tr-j} that every cycle in the graph must contain a tracker. This implies that graph induced by $S\setminus T'$ cannot contain a cycle. 
Thus, we have the following reduction rule.

\begin{Reduction Rule}
\label{red:no-cycle-in-X}
If the graph induced by $S\setminus T'$ is not a forest, reject the current guess for $T'$.
\end{Reduction Rule}

%Henceforth, we assume that the graph induced by $S\setminus T'$ is a forest.  

Now we define \textit{local source} and \textit{local destination}.
For a subgraph $G'\subseteq G$, and vertices $u,v\in V(G')$, $u$ is a local source and $v$ is a local destination if
\begin{enumerate}
\item there exists a path in $G$ from $s$ to $u$, say $P_{su}$,
\item there exists a path from $v$ to $t$, say $P_{vt}$,
\item $V(P_{su})\cap V(P_{vt})=\emptyset$, and
\item $V(P_{su})\cap V(G')=\{u\}$ and $V(P_{vt})\cap V(G')=\{v\}$.
\end{enumerate}

For a subgraph $G'$, we can check if a pair of vertices $a,b\in V(G')$ forms a local source-destination pair if there exists disjoint paths from $s$ to $a$ and $b$ to $t$ in the graph $G\setminus G'\cup\{a,b\}$, in quadratic time using the disjoint path algorithm from~\cite{KAWARABAYASHI2012424}. The concept of local source-destination pair has been used to obtain efficient algorithms for \tp (see~\cite{tr-j,ep-planar,iwoca}). If $u,v$ form a local source-destination pair for a subgraph $G'$, we refer to them as a local $s$-$t$ pair.
Next we recall the following lemmas (rephrased) from~\cite{tr-j}.

\begin{lemma}
\label{lemma:local-s-t}
In a reduced graph $G$, any induced subgraph $G'$ comprising of at least one edge contains a local source and a local destination.
\end{lemma}

\begin{lemma}
\label{lemma:subgraph}
In a subgraph $G'\subseteq G$, if all paths between a \lstpair cannot be tracked with at most $x$ trackers, then $G$ cannot be tracked with at most $x$ trackers.
\end{lemma}

Note that a subgraph can have more than one \lstpairs.
Now we can analyze subgraphs in $G$ and identify trackers with respect to the \lstpairs in that subgraph.
Next we give two rule that help mark some vertices as trackers, and thus reduce the number of unmarked vertices in a graph.

\begin{Reduction Rule}
\label{red:triangle-mark}
If $abc$ is a triangle in $G$ such that $a,b$ is a \lstpair for the triangle $abc$ and $c\notin S$, then mark $c$ as a tracker. While considering the disjoint version, if $c\in S\setminus T'$, then reject the current guess for $T'$.
\end{Reduction Rule}

\begin{lemma}
\label{lemma:red-triangle-mark}
Reduction Rule~\ref{red:triangle-mark} is safe and can be applied in polynomial time.
\end{lemma}
\begin{proof}
Observe that if $a,b$ form a \lstpair, then there exists a path from $s$ to $a$, say $P_{sa}$, that intersects with $abc$ only at $a$, and there exists a path from $b$ to $t$, say $P_{bt}$, that intersects with $abc$ only at $b$. Now consider the paths $P_1=P_{sa}\cdot (a,b) \cdot P_{bt}$ and $P_2=P_{sa}\cdot (a,c)\cdot (c,b)\cdot P_{bt}$. Here $c$ is the only vertex distinguishing between the paths $P_1$ and $P_2$, hence $c$ must be marked as a tracker. Since the disjoint tracking set problem requires to find a tracking set disjoint from $S-T'$, if $c$ is found to belong to $S-T'$, then the current guess for $T'$ is not correct and needs to be rejected. For applying the Reduction Rule, we consider all set of vertices of size three in $G$ in $\mathcal{O}(n^3)$ time, and we check if they form a triangle. For each triangle $abc$, we check if a pair of vertices among $a,b,c$ forms a \lstpair in $\mathcal{O}(n^2)$ time, then we mark the third vertex as a tracker if it satisfies our conditions. The total time taken is clearly polynomial in the size of graph $G$.
\qed
\end{proof}

%The above Reduction Rule helps in deriving the following one as well.

\begin{Reduction Rule}\label{red:clique}
If there exists a subgraph $G'\subseteq G$, such that $G'$ is a clique, with $a,b\in V(G')$ as a \lstpair for $G'$, then all vertices in $V(G')\setminus\{a,b\}$ need to be marked as trackers. Further, if $a,b$ is the only \lstpair for $G'$, then delete all vertices in $V(G')\setminus\{a,b\}$.
\end{Reduction Rule}

\begin{lemma}
\label{lemma:red-clique}
Reduction Rule~\ref{red:clique} is safe and can be applied in polynomial time.
\end{lemma}
\begin{proof}
Since $G'$ is a clique, edge $(a,b)\in E(G)$. Observe that all vertices in $V(G')\setminus\{a,b\}$ form triangles with the vertices $a$ and $b$. Due to Reduction Rule~\ref{red:triangle-mark}, all vertices in $V(G')\setminus\{a,b\}$ shall be marked as trackers.
We claim that the rule is safe. Suppose not. Let $G_1$ be the graph obtained from $G$ after deletion of all vertices in $V(G')\setminus\{u,v\}$. Then there exists two \stpaths, say $P_1,P_2$, in $G_1$ that contain the same sequence of trackers, or the deletion removes some untracked \stpaths from $G$. However, note that the deleted vertices can not participate in any \stpaths in $G_1$. Thus the only possibility of removal of some untracked \stpaths from $G$. Observe that all the deleted vertices in $V(G')\setminus\{u,v\}$ had already been marked as trackers. Thus, any \stpath formed due to the deleted vertices must have already been distinguished from other \stpaths in $G$. Hence the rule is safe. There can be at most $n$ components in $G-S$ and each component can have at most $n$ vertices. As mentioned earlier, in polynomial time, we can identify the \lstpairs in a component and mark all the remaining vertices in the component as trackers. If there exists only one \lstpair, in constant time, we can delete all the vertices other than $a$ and $b$ in the subgraph. Thus the rule is applicable in polynomial time.
\qed
\end{proof}

%Note that while applying Reduction Rule~\ref{red:triangle-mark}, if $c\in S-T'$, we can not mark it as a tracker. In such a case we reject the current guess for $T'$.
%Now we analyze how to solve the disjoint tracking set problem.

\subsection{Finding a Disjoint Tracking Set}

Let $A= S\setminus T'$ and $B=V \setminus S$. Note that our aim is to find a tracking set that is disjoint from $A$, i.e. $T\subseteq T'\cup B$.  
We first look at some structures induced by vertices in $S$ and $B$ that tries to force some vertices in $B$ as trackers and mark them.
% For the remaining vertices, we will consider all possible subsets, and then along with the vertices already added to $T$ verify if they form a valid tracking set using Theorem~\ref{theorem:np}.
We start with the following reduction rule. 

\begin{Reduction Rule}
\label{red:bad-guess-cycle}
Let $a,b \in T'$. If $a,b$ form a \lstpair for a subgraph $G'\subseteq G[\{a,b\}\cup A]$ (the subgraph induced by $A$ and $a, b$) and $G'$ induces a cycle, 
then reject $T'$, and move to the next guess.
\end{Reduction Rule}

\begin{lemma}
\label{lemma:red-bad-guess-cycle}
Reduction Rule~\ref{red:bad-guess-cycle} is safe and can be implemented in polynomial time.
\end{lemma}
\begin{proof}
Observe that if $a,b$ form a \lstpair for a subgraph $G'$ and $G'$ induces a cycle, then there exist two paths between $a$ and $b$ passing through $A$. Recall that we do not mark any trackers in $A$ as we assume that the tracking set $T$ intersects only with $T'$. Hence, we can not construct a tracking set for this situation, and we reject the current guess for $T'$.
In order to implement the reduction rule, we can run the algorithm for disjoint paths~\cite{KAWARABAYASHI2012424} in quadratic time, as explained before.
\qed
\end{proof}

%We start with the following reduction rule.

Since $S$ is a DCM, each vertex in $B$ is adjacent to at least $2$ vertices in $S$. We categorize the vertices in $B$ based on whether their neighbours lie in $A$ or $S-A$ as follows:
\begin{itemize}
\item $V_1$: The set of vertices that have at least two neighbours in $A$.
\item $V_2$: The set of vertices that have at least one neighbour in $A$ and at least one in $S-A$.
\item $V_3$: The set of vertices that have at least two neighbours in $S-A$.
\end{itemize}

%Here, vertices in $V_1$ may have more than two neighbours in $A$, but since we shall use just the fact that they have two neighbours in $A$ in our analysis, we have defined the vertex category accordingly. The same holds true for $V_2$ and $V_3$ as well. 
Observe that $B=V_1\cup V_2 \cup V_3$. There can be vertices that belong to more than one category, but as we shall see, this does not affect the outcome of the algorithm. Consider the case in which a pair of vertices $u,v\in S$ is adjacent to two vertices $w,x\in B$.  Observe that the vertices $u,v,w,x$ induce a $C_4$, say $C$. Due to Lemma~\ref{lemma:local-s-t}, there exists a \lstpair in the subgraph $C$. Now we analyze each of the above listed vertex sets in $B$, and we consider the possibility of each pair of vertices in a $C_4$ being an \lstpair for that $C_4$. Then we mark all those vertices as trackers that necessarily need to belong to $T$, and we bound the number of unmarked vertices in $V\setminus S$.
%We consider three possible cases based on whether $u,v$ belong to $A$ or $S\setminus A$, and $w,x\in B$, $w,x\notin T$.

\subsubsection{Bounding $V_1-T$}

Consider a set of vertices in $B$ that have two neighbours $u,v$ in $A$, i.e. $S\setminus T'$. See Figure~\ref{fig:c4-one}. Here, $u,v$ cannot be trackers, as $A\cap T=\emptyset$. 

\begin{figure}[ht]
\centering
\includegraphics[scale=0.4]{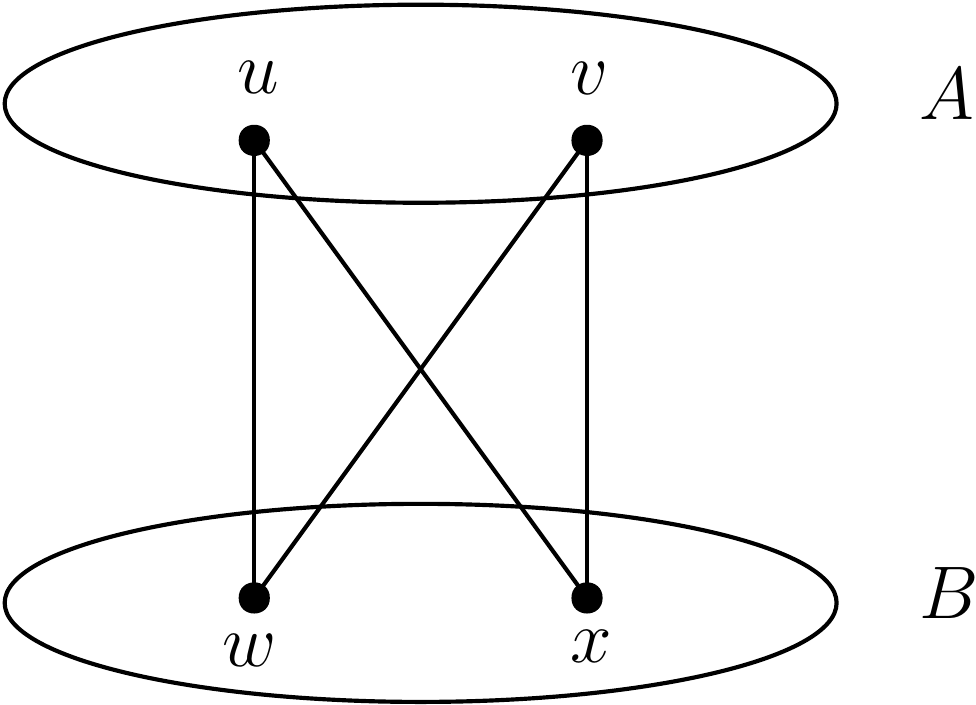} 
\caption{A pair of vertices in $A$ adjacent to two vertices in $B$ forms a $C_4$} 
\label{fig:c4-one}
\end{figure}

\begin{lemma}
\label{lemma:bounding-V1}
The number of vertices in $V_1-T$ can be bounded by $k\choose 2$.
\end{lemma}

\begin{proof}
We check for the possibility of each pair of vertices in $V(C)$ being a local source-destination pair in the following sequence:
\begin{enumerate}
\item If $\{w,x\}$ form a \lstpair:

Observe that there exists two paths between $w$ and $x$ in $C$, the first one passing through $u$, and the other passing through $v$. However, since both $u,v\in A$, we can not mark them as trackers. Thus in this case we can not find a tracking set for the graph. So we move on to the next guess for $A$.

\item If $u,w$ form a \lstpair: 

If $(u,v)\in E(G)$, then it would have led to rejection of the current guess for $T'$ due to Reduction Rule~\ref{red:triangle-mark}.
Observe that there exists two paths between $u$ and $w$ in $C$: first the edge $(u,w)$, and second the path $u\cdot x\cdot v\cdot w$. Since the path $(u,w)$ does not contain a vertex other than $u$ and $w$, there must be a tracker on the path $u\cdot x\cdot v\cdot w$. Since $v\in A$, we can not mark it as a tracker. Thus, $x$ necessarily has to be marked as a tracker.
Note that this case is symmetric to the cases when $\{u,x\}$, or $\{v,w\}$, or $\{v,x\}$ form a local source-destination pair. In all these cases, we mark the vertex that belongs to $B$ and is not part of the local source-destination pair, as a tracker.

\begin{figure}[ht]
\centering
\includegraphics[scale=0.4]{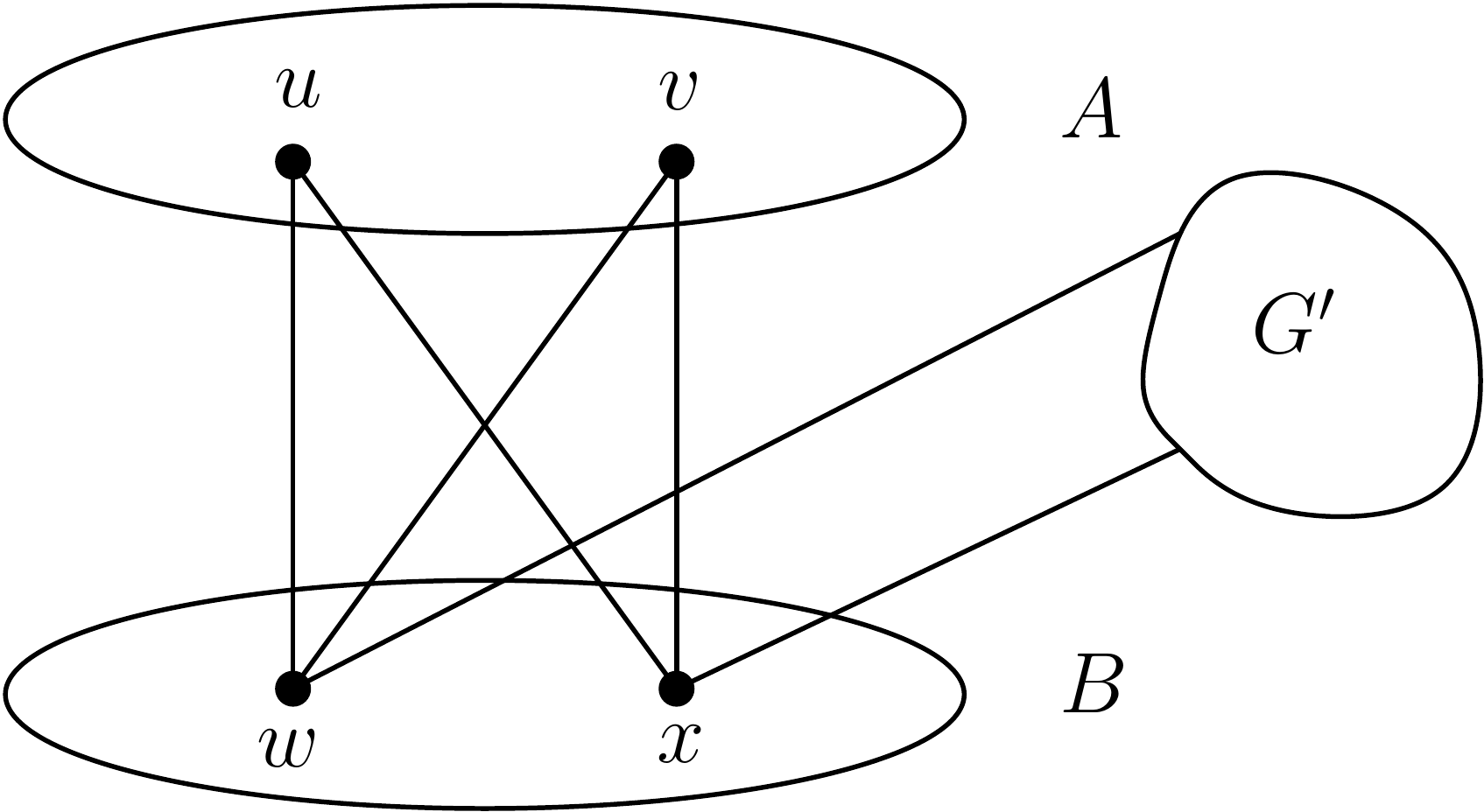} 
\caption{When $u,v$ are a local source-destination pair} 
\label{fig:c4-arb-one}
\end{figure}

\item If $\{u,v\}$ form a \lstpair:

If $(u,v)\in E(G)$, then due to Reduction Rule~\ref{red:triangle-mark} $w,x \in T$.
Else, observe that there exists two paths between $u$ and $v$ in $C$: first one passing through $w$, and the other passing through $x$. If $deg(w)=deg(x)=2$, we arbitrarily mark one of them as a tracker. Else, note that we reach this case when neither $w$ nor $x$ are a local source or destination for $C$. Hence, the degree of both $w$ and $x$ is at least $3$. Suppose not, and let $deg(x)=2$ but $deg(w)\geq 3$. Let $y\in N(w)\setminus\{u,v\}$. Due to Reduction Rule~\ref{red:stpath}, edge $(w,y)$ must participate in an \stpath. Then there must exist a path from $w$ to $t$ that does not include $v$, and hence it is not possible that $w$ is not a local destination. Thus the graph is similar to the one shown in Figure~\ref{fig:c4-arb-one}. Observe that both $\{w,x\}$ and $\{x,w\}$ form \lstpairs for the subgraph $G'$, and hence they do not play a role in tracking the subpaths inside $G'$. Thus we can arbitrarily mark either $w$ or $x$ as a tracker.

\end{enumerate}

After applying the above steps, there does not exist a $C_4$ induced by a pair of vertices in $A$ and a pair of vertices in $B$ such that both the vertices from $B$ are unmarked. Now each pair of vertices in $A$ is adjacent to at most one unmarked vertex in $B$. Since $|A|\leq k$, the number of vertices in $V_1-T$ i.e. unmarked vertices in $V_1$ is at most $k\choose 2$. 
\qed
\end{proof}

\subsubsection{Bounding $V_2-T$}

Here we consider the set of vertices in $B$ that have one neighbour in $S\setminus A$ and one neighbour in $A$. Since $u\in S\setminus A,v\in A$, $u$ is already marked as a tracker and $v$ cannot be a tracker.

\begin{lemma}
\label{lemma:bounding-V2}
The number of vertices in $V_2-T$ can be bounded by $2{k\choose 2}$.
\end{lemma}

\begin{proof}
Here we consider the scenario when a pair of vertices in $S$ is adjacent to three unmarked vertices in $B$.
We call the induced $K_{2,3}$ as $G'$. See Figure~\ref{fig:c4-one-tracker-3b}.

\begin{figure}[ht]
\centering
\includegraphics[scale=0.35]{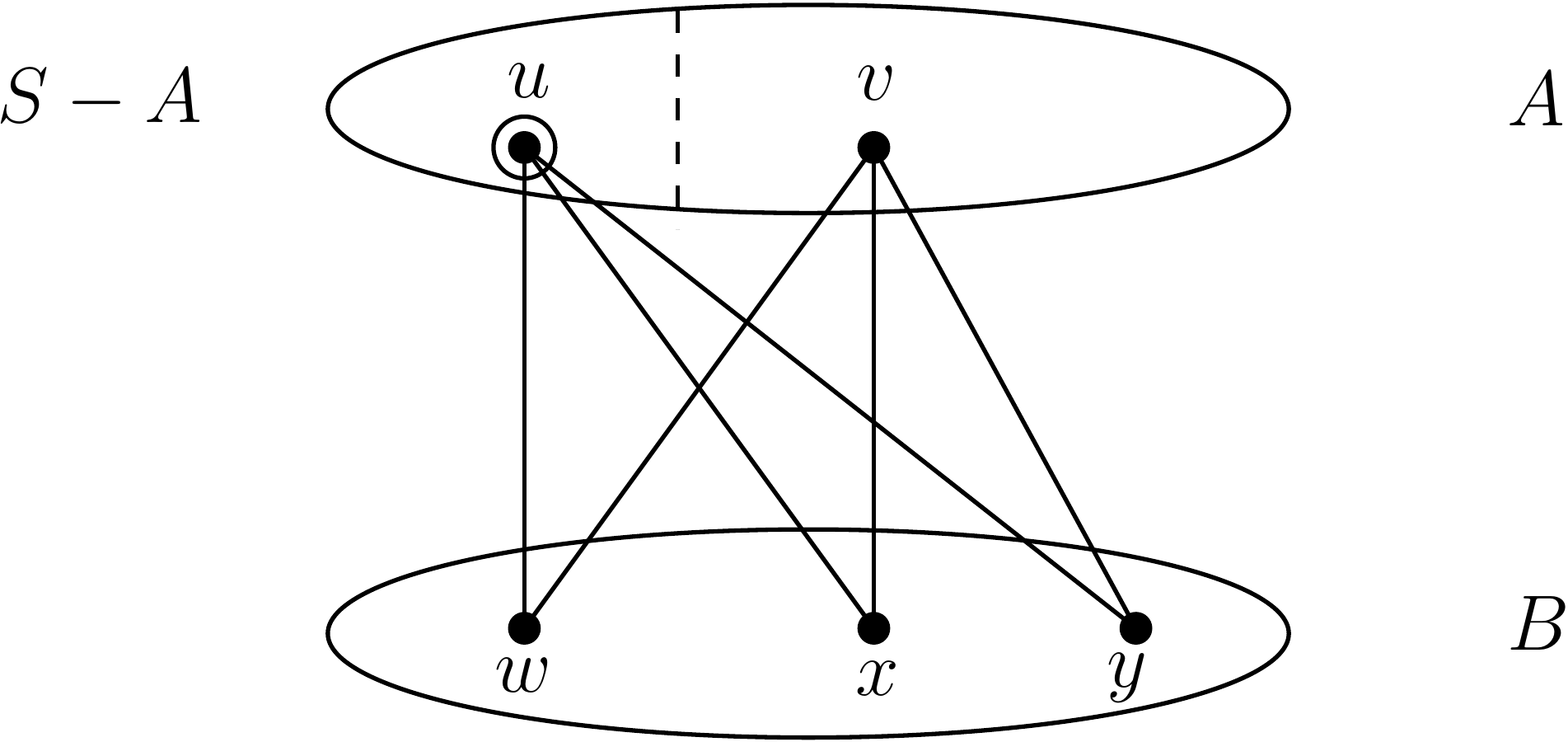} 
\caption{When $u,v$ are a local source-destination pair} 
\label{fig:c4-one-tracker-3b}
\end{figure}

We consider the different cases based on local source-destination for $G'$ in the following sequence.
\begin{enumerate}
\item If $w,x$ form a \lstpair:

Suppose $(u,v)\in E(G)$. Then observe the two paths between $w$ and $x$: first being $w\cdot u\cdot x$ and second being $w\cdot u\cdot v\cdot x$. Note that these paths differ only in the vertex $v$ which can not be marked as a tracker since $v\in A$. Hence, we reject this choice of $A$, and move on to the next choice.

Else $(u,v)\notin E(G)$. Consider the paths $P_1=w\cdot u\cdot x$ and $P_2=w\cdot v\cdot y\cdot u\cdot x$. Either $v$ or $y$ must be marked as a tracker in order to distinguish between $P_1$ and $P_2$. Since $v\in A$, we can not mark it as a tracker. Hence, we mark $y$ as a tracker. Note that this case is similar to the case when $\{w,y\}$ or $\{x,y\}$ form a local source-destination pair. In these cases we mark the vertex in $B\cap G'$ that is not part of the local source-destination pair as a tracker.

\item If $u,w$ form a \lstpair: 

Suppose $(u,v)\in E(G)$. Then observe the two paths between $u$ and $w$: first being the edge $(u,w)$ and second being the path $u\cdot v\cdot w$. Note that these paths differ only in the vertex $v$ which can not be marked as a tracker since $v\in A$. Hence, we reject this choice of $A$, and move on to the next choice.

Else $(u,v)\notin E(G)$. Observe that in $C$, there exists three paths between $u$ and $w$: $P_1=u\cdot w$, $P_2=u\cdot x\cdot v\cdot w$, $P_3=u\cdot y\cdot v\cdot w$.  Since $v$ can not be marked as a tracker, $x,y$ necessarily have to be marked as trackers.
Note that this case is symmetric to the cases when $\{u,x\}$ or $\{u,y\}$ form local source-destination pairs. In these cases, we mark the vertices in $B\cap G'$ that are not part of the local source-destination pair as trackers.

\item If $v,w$ form a \lstpair: 

Consider the paths, $P_1=v\cdot x\cdot u\cdot w$ and $P_2=v\cdot y\cdot u\cdot w$. In order to distinguish between $P_1$ and $P_2$ we need a tracker on at least one of the vertices among $x,y$. Hence, we arbitrarily mark $x$ or $y$ as a tracker. 
Note that this case is symmetric to the case when $\{v,x\}$ or $\{v,y\}$ form local source-destination pairs. In these cases also we arbitrarily mark one vertex as a tracker amongst the ones that are not a local source or destination.

Observe that the only additional possibility is that here in $G'$, $u,v$ might also be forming another local source-destination pair. However, marking arbitrarily one among $x,y$ as a tracker does not violate our analysis for that case as explained in next case.

Note that if $(u,v)\in E(G)$, then in order to distinguish $P_1$, $P_2$, and $P_3=v\cdot u\cdot w$, both $x$ and $y$ necessarily need to be marked as trackers.

%\begin{figure}[ht]
%\centering
%\includegraphics[scale=0.4]{c4-arb-one.eps} 
%\caption{When $u,v$ are a local source-destination pair} 
%\label{fig:c4-arb-one}
%\end{figure}

\item If $u,v$ form a \lstpair:

Observe that there exists $3$ vertex disjoint paths between $u$ and $v$. 
If $deg(w)=deg(x)=deg(y)=2$, we arbitrarily mark two among $w,x,y$ as trackers. Suppose not.
Then degree of at least two among $w,x,y$ is greater than or equal to three. Without loss of generality suppose $deg(y)=2$ and $deg(w)\geq 3$, $deg(x)\geq 3$. Then the case for $w$ and $x$ is similar to that depicted in Figure~\ref{fig:c4-arb-one}. Hence we arbitrarily mark one among $w,x$, say $w$, as a tracker.

Note that if $(u,v)\in E(G)$, then all three among $x,y,z$ necessarily need to be marked as trackers.

\end{enumerate}

After applying the above steps, each pair of vertices in $S$, such that one of the vertices is in $S\setminus A$ and another is in $A$, is adjacent to at most two unmarked vertices in $B$. Since $|A|\leq k$, the number of vertices in $V_2-T$ i.e. unmarked in $V_2$ is at most $2 {k\choose 2}$. 
\qed
\end{proof}

\subsubsection{Bounding $V_3-T$}

Here we consider the set of vertices in $B$ that have two neighbours in $S\setminus A$. Since $u,v\in S\setminus A$, both $u$ and $v$ are already marked as trackers. Since we have already analyzed the vertices in $B$ that have at least one neighbour in $A$, here we restrict ourselves to only those vertices of $B$ that are adjacent to only $S-A$. Let $G'$ be the subgraph induced by $(S\setminus A)\cup V_3$. 

%Note that if $V_3$ is an independent set, for any path in $G'$, at least every alternate vertex is a tracker. Hence, if a pair of paths between a local source and destination in $G'$, has the same sequence of trackers, then the cycle induced by the distinct portions of these paths is of length at most four. 

%Further note that the vertices in $V_3$ are connected to the remaining graph only through the vertices in $S\setminus A$, which are all marked as trackers.

If a pair of vertices $u,v\in S\setminus A$ are adjacent to a pair of vertices $w,x\in B$, they induce a $C_4$, say $C$. We create an empty set $V_3'$, which will be used to identify those vertices of $V_3$ which might later be needed to be marked as trackers. We check for the possibility of each pair of vertices in $V(C)$ being a \lstpair in the following sequence:
\begin{enumerate}

\item If $w,x$ form a \lstpair: Here $u,v$ already serve as trackers to distinguish the two paths in $C$ between $w$ and $x$. Even if $(u,v)\in E(G)$, we do not need any more trackers.
\item If $u,w$ form a \lstpair: Here $v$ already serves as a tracker to distinguish the paths $u.w$ and $u.x.v.w$. If $(u,v)\in E(G)$, then we mark $x$ as a tracker to distinguish the paths $u.v.w$ and $u.x.v.w$.

\item If $u,v$ form a \lstpair: If $(u,v)\in E(G)$, then we need to mark both $w,x$ as trackers, in order to distinguish the paths $u.v$, $u.w.v$ and $u.x.v$. If $(w,x)\in E(G)$, then also we need to mark both $w,x$ as trackers, in order to distinguish the paths $u.w.v$, $u.x.w.v$ and $u.w.x.v$. If $(u,v),(w,x)\notin E(G)$, we arbitrarily mark one among $w,x$ as a tracker, and add the other vertex to $V_3'$. 
%
%If $(u,v)\in E(G)$ then $w,x,y,z$ need to be marked as trackers. Else, at least three among $w,x,y,z$ need to be marked as trackers.

\end{enumerate}

We add to $V_3'$ those vertices from $V_3$ that are adjacent to a unique pair of vertices from $T'$.
Thus $|V_3'|\leq {k\choose2}$. 
If $V\setminus S$ (hence $V_3$) is an independent set, for any path in $G'$, at least every alternate vertex is a tracker. Hence, if a pair of paths between a local source and destination in $G'$, has the same sequence of trackers, then the cycle induced by the distinct portions of these paths is of length at most four. Else due to Reduction Rule~\ref{red:clique}, each component in $G-S$ has at most two unmarked vertices. Hence, if $u,v\in T'$, then there exists a path of length at most two between $u$ and $v$, in $G(V_3\cup\{u,v\})$, if all vertices in the path are unmarked. Thus if two such paths exist between a pair of vertices in $T'$ that forms a \lstpair, then we mark a vertex as tracker. In ${k\choose 2}.n^{\Oh(1)}$ time we can find if two such paths exist between a pair of vertices in $T'$ that forms a \lstpair. We arbitrarily mark one vertex on one of the paths as a tracker, and the two vertices from the other subpath are added to $V_3'$. Hence $|V_3'|\leq 2{k\choose 2}$.

Next we prove that we need not consider vertices from $V_3\setminus V_3'$ as trackers.

\begin{lemma}
\label{lemma:V_3-V_3'}
There exists a $k$ sized tracking set for $G$ if and only if there is one that is a subset of $V_1\cup V_2\cup V_3'\cup T'$.
%Vertices in $V_3\setminus V_3'$ need not be marked as trackers.
\end{lemma}

\begin{proof}
We argue that an optimal tracking set need not contain vertices from $V_3\setminus V_3'$. Suppose the lemma does not hold. Then there exists a vertex $a\in V_3\setminus V_3'$ such that if $a$ is not marked as a tracker, and there exists two \stpaths, say $P_1,P_2$, that contain the same sequence of trackers. Observe that all neighbours of $a$ are marked as trackers. Hence, there must be sub-paths of $P_1,P_2$ that lie between a pair of vertices in $N(a)$, say $x,y$, such that $a$ belongs to one of them and does not belong to another. Without loss of generality, let $P_1'\subseteq P_1$ be the path $x\cdot a\cdot y$ and $P_2'\subseteq P_2$ be a subpath of $P_2$ lying between the vertices $x,y$ such that it does not pass through $a$. Note that $x,y$ act as a local source-destination pair for the graph induced by $V(P_1')\cup V(P_2')$. If $a$ is the unique vertex adjacent to the pair of vertices $x,y$, then $a\in V_3'$. Else, the first possibility is that there exists another vertex, say $b$, such that $b\in B\cap N(x)\cap N(b)$. In this case, vertices $a,b,x,y$ induce a $C_4$ and $x,y$ serve as a local source-destination pair for this $C_4$. Due to the steps applied above, one of $a,b$ must have already been marked as a tracker, and the other one must have been added to $V_3'$. This contradicts the assumption that $a\notin T\cup V_3'$. The second possibility is that $P_2'$ passes through $A$ or $A\cup B$, since all other vertices in $S-A$, i.e. $T'$, are already trackers. Note that due to Reduction Rule~\ref{red:bad-guess-cycle}, this $P_2'$ can not pass through just $A$. Further, if $P_2'$ passes through $A\cup B$, then there exists a pair of vertices, say $u,v\in A$ such that $P_2'$ passes through $u$, then a vertex, say $w$, and then through $v$. Note that in this case if $w$ is not already a tracker, then it is a unique vertex that is adjacent to $u,v$, and hence belongs to $V_1$, and shall be eventually marked as a tracker when we consider all possible guesses for $T\cap (V_1\cup V_2\cup V_3')$.
\qed
\end{proof}

%
%
%We consider all possible subsets of these $\Oh(k^2)$ vertices in $B$, include them with the already marked trackers to form a tracking set, and check the validity of that tracking set. Thus the algorithm takes $2^{\Oh(k^2)}.n^{\Oh(1)}$ time.

Now each pair of vertices in $S$ is adjacent to at most one unmarked vertex or a vertex from $V_3'$. Note that $|V_1\cup V_2\cup V_3'|\leq 5{k\choose 2}$.
%
%there does not exist a $C_4$ induced by a pair of vertices in $S\setminus A$ and a pair of vertices in $B$ such that neither of the vertices from the pair from $B$ are marked as a tracker. Hence, each pair of vertices in $A$ is adjacent to at most one vertex not marked as a tracker from $B$. Since $|A|\leq k$, then number of vertices not marked as trackers in $B$ is at most $k\choose 2$. 
We consider all possible subsets of these $\Oh(k^2)$ vertices, include them with the already marked trackers to form a tracking set, and check the validity of that tracking set using Theorem~\ref{theorem:np}. Thus the overall algorithm takes $2^{\Oh(k^2)}.n^{\Oh(1)}$ time. 
While solving the disjoint problem, for each new guess for $T'$, if the size of the tracking set found is less than that of the tracking set found with respect to the previous guess, we discard the previous disjoint solution, and retain the new one, else we discard the current one. Hence we have the following theorem.

\begin{theorem}
\label{theorem:tp-dcm}
For a graph $G$ with a known \textit{dual connected modulator} of size $k$, \tp can be solved in $2^{\Oh(k^2)}.n^{\Oh(1)}$ time.
\end{theorem}

\section{Parameterization by Vertex Cover}
\label{sec:vc}

In this section we give an FPT algorithm for \textsc{Tracking Paths} when the parameter $k$ is the size of a known vertex cover for the given graph.

\defproblem{\tp/\textsc{Vertex Cover} $(G,s,t,S,k)$}{An undirected  graph $G=(V,E)$ with terminal vertices $s$ and $t$, and a vertex cover $S\subseteq V(G)$ for $G$, such that $|S|=k$.}
{Find a minimum cardinality tracking set  $T$ for $G$.}
\medskip

We start by applying Reduction Rules~\ref{red:stpath} and \ref{red:no-deg-one}. Observe that these rules delete vertices/edges from the input graph, and hence do not tamper with the vertex cover $S$. Since Reduction Rule~\ref{red:degree-two} introduces a new edge in the graph, we tweak the rule in order to maintain that $S$ is a vertex cover of size at most $k$.

\begin{Reduction Rule}
\label{red:degree-two-vc}
Let $u,v\in V(G)$ such that $deg(u)=deg(v)=2$, $N(v)=\{u,w\}$, then delete $v$ and introduce an edge between $u$ and $w$. If $u,w\notin S$, set $S=S\cup \{u\}$.
\end{Reduction Rule}

Note that if $u,w\notin S$, then $v$ necessarily belongs to $S$. Since deletion of $v$ reduces the size of $S$ by one, we can safely add $u$ (or $w$) to $S$ without increasing the value of $k$.

%We retain the notations of $G,S,k$ to denote the preprocessed graph, vertex cover of the preprocessed graph and the size of $S$, respectively.
Now observe that due to Reduction Rules~\ref{red:stpath} and \ref{red:degree-two-vc}, each vertex in $G-S$ has at least two neighbours in $S$. Hence $S$ is a DCM for $G$. Thus we can apply the algorithm for DCM to derive an FPT algorithm for parameterization by the size of a vertex cover for a graph. Hence we have the following theorem.

\begin{theorem}
\label{theorem:tp-vc}
For a graph $G$ with a known vertex cover of size $k$, \tp can be solved in $2^{\Oh(k^2)}.n^{\Oh(1)}$ time.
\end{theorem}

%Given a graph $G=(V,E)$, let $E_m\subseteq E$ be a matching for $G$. Observe that $G(E\setminus E_m)$ is an edgeless graph. Further due to Reduction Rule~\ref{red:d
%\begin{corollary}

\section{Parameterization by Cluster Vertex Deletion set}
\label{sec:cvd}

In this section we give an FPT algorithm for \textsc{Tracking Paths} when the parameter $k$ is the size of a cluster vertex deletion set for the given graph.

\defproblem{\tp/\textsc{Cluster Vertex Deletion Set} $(G,s,t,S,k)$}{An undirected  graph $G=(V,E)$ with terminal vertices $s$ and $t$, and a cluster vertex deletion set $S\subseteq V(G)$ for $G$, such that $|S|=k$.}
{Find a minimum cardinality tracking set  $T$ for $G$.}
\medskip

First we apply the Reduction Rules~\ref{red:stpath}, \ref{red:no-deg-one} and \ref{red:degree-two} as explained below.
\begin{itemize}
\item
Reduction Rule~\ref{red:stpath}: Observe that $G-S$ is a cluster graph. Thus for each component in $G-S$, either all the vertices participate in an \stpath, or none of them do. Thus, when we apply Reduction Rule~\ref{red:stpath}, it might lead to deletion of some vertices/edges from $S$, or some components from $G-S$. Note that none of these operation affect the properties of $S$ or $G-S$.

\item
Reduction Rule~\ref{red:no-deg-one}: The application of this rule may delete vertices/edges from $S$ and/or single vertex components from $G-S$. Observe that this does not affect the properties of $S$ or $G-S$.

\item
Reduction Rule~\ref{red:degree-two}: Let $u,v\in V(G)$ be two vertices such that $deg(u)=deg(v)=2$. If both $u,v\in S$, we can apply the rule and its does not affect properties of $S$ or $G-S$. Consider the case in which one vertex among $u,v$ belongs to $S$ while the other belongs to $G-S$. In such a case, we necessarily delete the vertex that belongs to $G-S$. Suppose $u\in S$ and $v\in G-S$. Then $v$ must belong to a component in $G-S$ that comprises of only a single edge. After application of the reduction rule, the component of $v$ shall comprise of only a single vertex. Observe that this does not affect the properties of $S$ and $G-S$.
\end{itemize}

\noindent
We also apply Reduction Rules~\ref{red:triangle-mark} and ~\ref{red:clique} to mark required vertices in $G-S$ as trackers. Note that while applying all above reduction rules, it has been maintained that $G-S$ is a cluster graph and $|S|\leq k$.

Next, we try to mark as many vertices as possible as trackers in $G-S$, such that for the unmarked vertices $S$ is a DCM. We create two sets $X=Y=\emptyset$. We use $X$ to maintain the unmarked vertices in $G-S$, while ensuring that they have two neighbours each in $S$, and we use $Y$ to maintain some other vertices that might need to be marked as trackers.
Now we identify \lstpairs in each component. Due to Lemma~\ref{lemma:local-s-t}, each component (having at least one edge) in $G-S$ has at least one \lstpair. After the application of Reduction Rule~\ref{red:clique}, for each \lstpair in a component, all the remaining vertices shall be marked as trackers.

\begin{corollary}
\label{cor:two-unmarked}
After application of Reduction Rule~\ref{red:clique} at most two vertices in each component of $G-S$ are left unmarked.
\end{corollary}

First, we consider the components in $G-S$ that contain $s$ or $t$, or both. Let $G'$ be a component in $G-S$ such that $s\in V(G')$ ($t\in V(G')$). Due to Lemma~\ref{lemma:local-s-t}, $G'$ contains a \lstpair, say $s,a$ ($b,t$). Due to Reduction Rule~\ref{red:clique}, all vertices in $V(G')\setminus\{s,a\}$ ($V(G')\setminus\{b,t\}$) shall be marked as trackers. If $deg_S(a)\geq 2$ ($deg_S(b)\geq 2$), we add $a$ ($b$) to the set $X$, else we add $a$ ($b$) to the set $Y$. Note that $|Y|\leq 2$. Henceforth, by `components' we mean components in $G-S$, and we assume that none of the components in $G-S$ contain $s$ or $t$.
Now we analyze different types of components in $G-S$ based on their sizes.

\subsection{Components with at least three vertices}

Due to Lemma~\ref{lemma:local-s-t}, each component has a \lstpair. Since we already analyzed the components that contain $s$ or $t$, a \lstpair in a component necessarily needs to be adjacent to $S$ in order to connect with rest of the graph. Thus, each component that does not contain $s,t$, has at least two neighbours in $S$. We consider different cases based on the number of vertices in each component of $G-S$ that have neighbours in $S$. %Note that each component has at least three vertices.

\subsubsection{Components with exactly two vertices with neighbours in $S$}
\label{subsec:two-vertices-with-nbrs}

%Let $G'$ be a component in $G-S$ that has exactly two neighbors in $S$. See Figure~\ref{fig:two-neighbor}.

%\vspace{-0.5cm}
\begin{figure}[ht]
\centering
\includegraphics[scale=0.27]{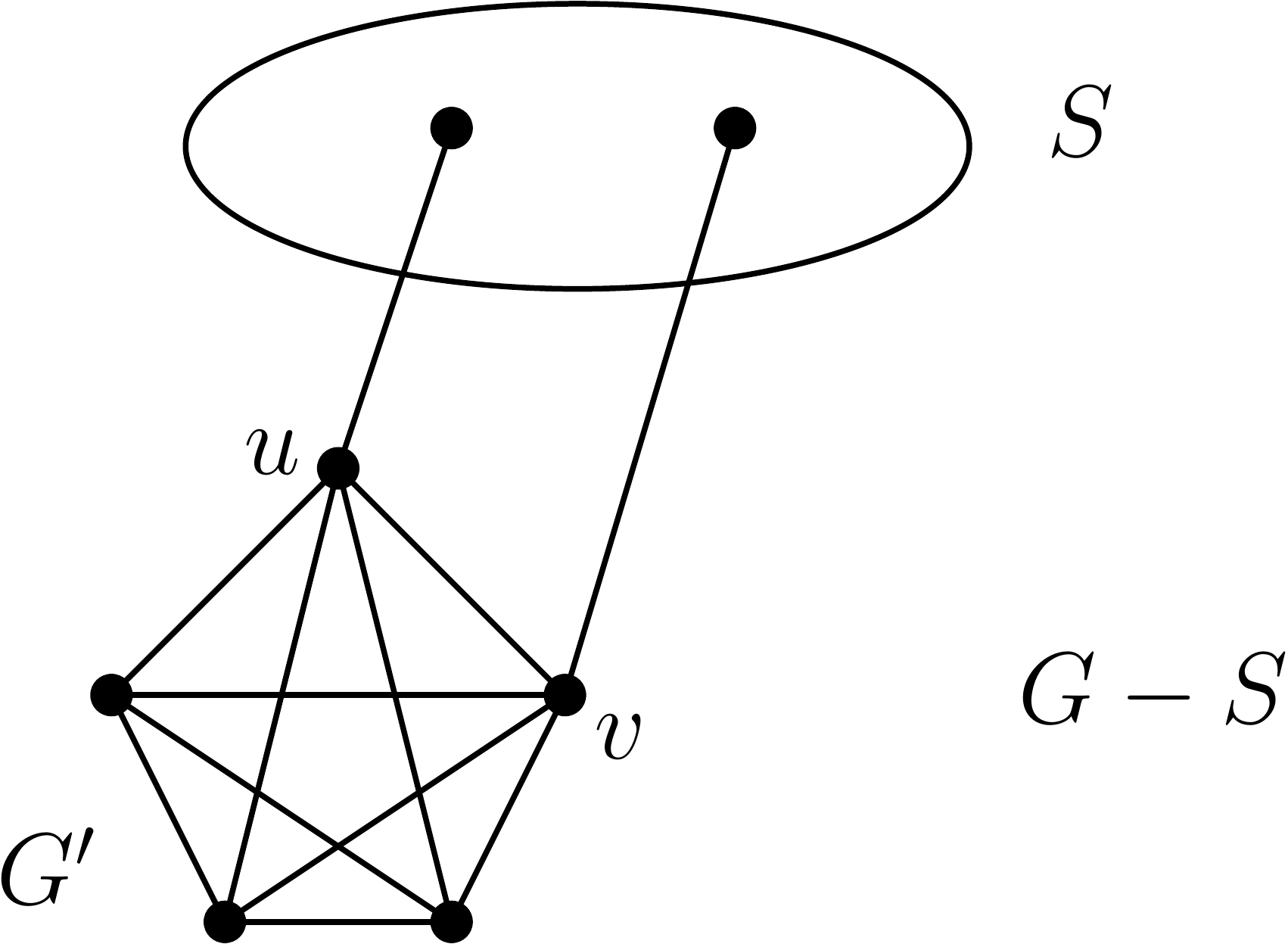} 
\caption{Component in $G-S$ with only two neighbours in $S$} 
\label{fig:two-neighbor}
\end{figure}
%\vspace{-0.5cm}

Let $G'\subseteq G-S$ be a component with exactly two vertices, say $u,v$, that have neighbours in $S$. See Figure~\ref{fig:two-neighbor}. 
Due to Lemma~\ref{lemma:local-s-t}, $G'$ has at least one \lstpair and since $s,t\notin V(G')$, $u,v$ shall form an \lstpair for $G'$ as these are the only vertices that connect $G'$ with rest of the graph. Due to application of Reduction Rule~\ref{red:clique}, all vertices in $V(G')\setminus\{u,v\}$ shall have been deleted. Thus if a component has only two neighbours in $S$, it can consist of at most two vertices.
The analysis for such components is explained in Section~\ref{subsec:single-vertex-edge-component}.

\subsubsection{Components with three or more vertices with neighbours in $S$}

Let $G'$ be a component in $G-S$. Due to Reduction Rule~\ref{red:clique}, if there are two disjoint \lstpairs, then all vertices in that component shall be marked as trackers. We need not analyze such components further. Henceforth, we assume that if a component in $G-S$ has more than one \lstpairs, then these pairs overlap. Consider a component $G'$ in $G-S$, with two \lstpairs. Since all \lstpairs overlap, at most three vertices in $G'$, say $a,b,c$, form these two \lstpairs.

\begin{figure}[ht]
\centering
\includegraphics[scale=0.37]{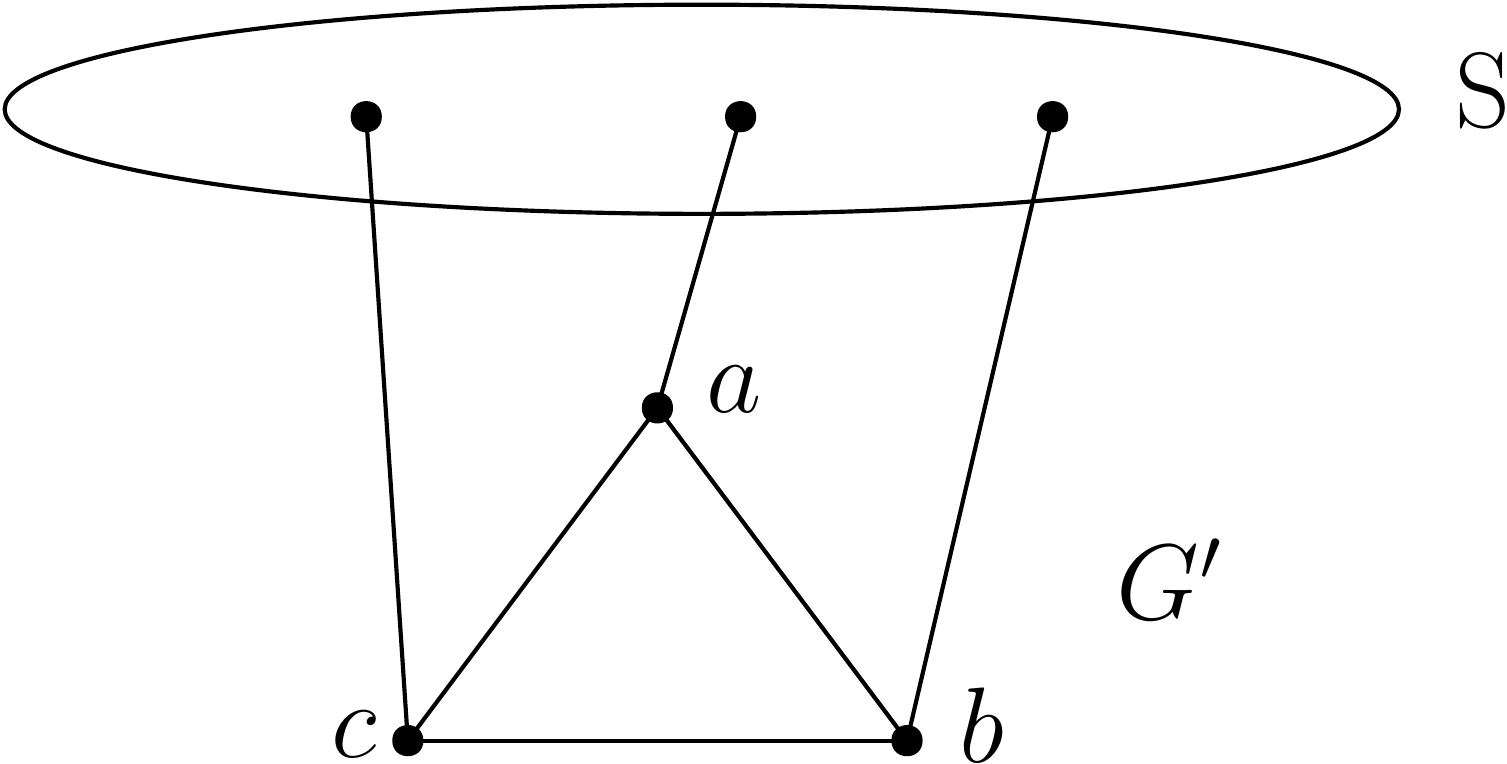} 
\caption{Component in $G-S$ with two \lstpairs} 
\label{fig:two-st-pairs}
\end{figure}

%Let $a,b,c\in V(G')$ be the vertices that form \lstpairs. 
Consider the case in which $a,b$ is a \lstpair and $b,c$ is another \lstpair. See Figure~\ref{fig:two-st-pairs}.
 Observe that both $c$ and $a$ shall be marked as trackers due to Reduction Rule~\ref{red:clique}. %Thus when two \lstpairs overlap, at most one vertex is left unmarked. 
Since here all \lstpairs in a component overlap, at most one vertex in the component will be left unmarked. If such unmarked vertices have at least two neighbours in $S$, we add them to $X$. Else, we prove that they need not be part of an optimum tracking set.

\begin{lemma}
\label{lemma:single-unmarked}
If there exists a component $G'\subseteq G-S$ with only one unmarked vertex $v$, and $deg_S(v)=1$, then such a vertex need not be marked as a tracker.
\end{lemma}

\begin{proof}
Let $G'\subseteq G-S$ be a component, with only one marked vertex, say $v$, such that $deg_S(v)=1$. Let $N_S(v)=\{u\}$. Observe that all paths passing through $v$, first traverse through $u$. We claim that $u$ can be replace $v$ as a tracker if needed. If there does not exist an optimum tracking set that contains $v$, then the claim holds. Else, we claim that an optimum tracking set $T$ containing $v$ can be replaced by $T\cup\{u\}\setminus\{v\}$. Suppose not. Then there exists at least two \stpaths in $G$, say $P_1,P_2$, that can not be tracked by replacing $v$ by $u$. Note that $v$ is part of some \lstpair for $G'$, else it would have been marked as a tracker. Without loss of generality let $v$ be a local source for $G'$. All vertices in $V(G')\setminus\{v\}$ and $u$ are already marked as trackers. Thus one among $P_1,P_2$ contains $v$ while the other does not. Let $v\in V(P_1)$ and $v\notin V(P_2)$. Since $P_1,P_2$ have the sequence of trackers, the only possibility here is that there exists and edge between $u$ and another vertex, say $w$, in $G'$. However, in such a case, $u,w$ form a \lstpair for the triangle $uvw$. Then due to application of Reduction Rule~\ref{red:triangle-mark} for vertices in $G-S$, $v$ would have already been marked as a tracker. This contradicts the assumption that $v$ is an unmarked vertex. Note that the other possibility of $v$ serving as a tracker by differentiating between sequence of vertices in two \stpaths is also not applicable here. This is due of the fact that all \stpaths that pass through $u$ and then $v$, also pass through other vertices of $G'$. Hence the vertices of $V(G')\setminus\{v\}$ along with $u$ can help distinguish \stpaths based on vertex sequences.
Thus, while considering all unmarked vertices that might be useful to serve as trackers, we can ignore $v$. 
\qed
\end{proof}

Now we are left with components that have two unmarked vertices. Clearly a pair of unmarked vertices in a component shall be a \lstpair for that component. %Note that such components do have more than two vertices with neighbours in $S$.%, else they would have been reduced due to Reduction Rule~\ref{red:two-nbr-comp}. 
Let $G'$ be a component with only one \lstpair, say $a,b$, but more than two vertices with neighbours in $S$. Consider the following cases:
\begin{enumerate}
\item Both $a$ and $b$ have at least two neighbours in $S$: We add both $a,b$ to $X$.

\item One among $a,b$, say $b$, has two neighbours in $S$, while $a$ has only one neighbour in $S$: We add $b$ to $X$. If $b$ is marked as a tracker while application of the DCM algorithm, due to Lemma~\ref{lemma:single-unmarked}, $a$ need not be considered as a tracker, and hence can be ignored. Else, if $b$ is left unmarked in the DCM algorithm, we can account for $a$ by doubling the bound obtained for unmarked vertices.

\item Both $a$ and $b$ have only one neighbour in $S$: Let $c\in N_S(a)$ and $d\in N_S(b)$. We introduce an additional vertex $v_{ab}$, and introduce edges $(v_{ab},c)$ and $(v_{ab},d)$ to $E(G)$. We also add $v_{ab}$ to $X$. If $v_{ab}$ is eventually part of a solution, then we can arbitrarily mark either $a$ or $b$, say $a$, as a tracker. Not that in order to distinguish paths by their vertex sequences, we can mark $c$ along with $a$, thus ruling out the necessity of marking both $a$ and $b$ as trackers.
\end{enumerate}

%We can merge all the marked vertices into a single marked vertex, and reduce the component to an edge with one marked and one unmarked vertex.

\subsection{Single vertex and single edge components}
\label{subsec:single-vertex-edge-component}

If a component in $G-S$ consists of a single vertex $v$, then due to Reduction Rules~\ref{red:stpath} and \ref{red:no-deg-one}, $v$ has at least two neighbours in $S$. We include $v$ in $X$, i.e. $X=X\cup \{v\}$.

If a component consists of a single edge $(a,b)\in E(G-S)$, due to Reduction Rules~\ref{red:stpath} and \ref{red:no-deg-one}, both $a,b$ have a neighbour in $S$. Due to Reduction Rule~\ref{red:degree-two}, it is not possible that $deg(a)=deg(b)=2$. Thus, at least one vertex in each single edge component has at least two neighbours in $S$. If both $a,b$ have two neighbours each in $S$, then we set $X=X\cup\{a,b\}$. 
%Henceforth, we assume that each edge component in $G-S$ has only one vertex with two or more neighbours in $S$.
Suppose $(a,b)$ is an edge component and $b$ has two neighbours in $S$, while $a$ has only one neighbour in $S$.
Observe that if while applying the algorithm for DCM, $b$ is among the $\mathcal{O}(k^2)$ bounded vertices that are left unmarked, then we can simply double this bound in order to account for vertices like $a$ (Note that this does not change the asymptotic bound of the running time). Else, if while applying the algorithm for DCM, if $b$ is marked as a tracker, we show that an optimum tracking set for $G$ need not contain $a$.

\begin{figure}[ht]
\centering
\includegraphics[scale=0.34]{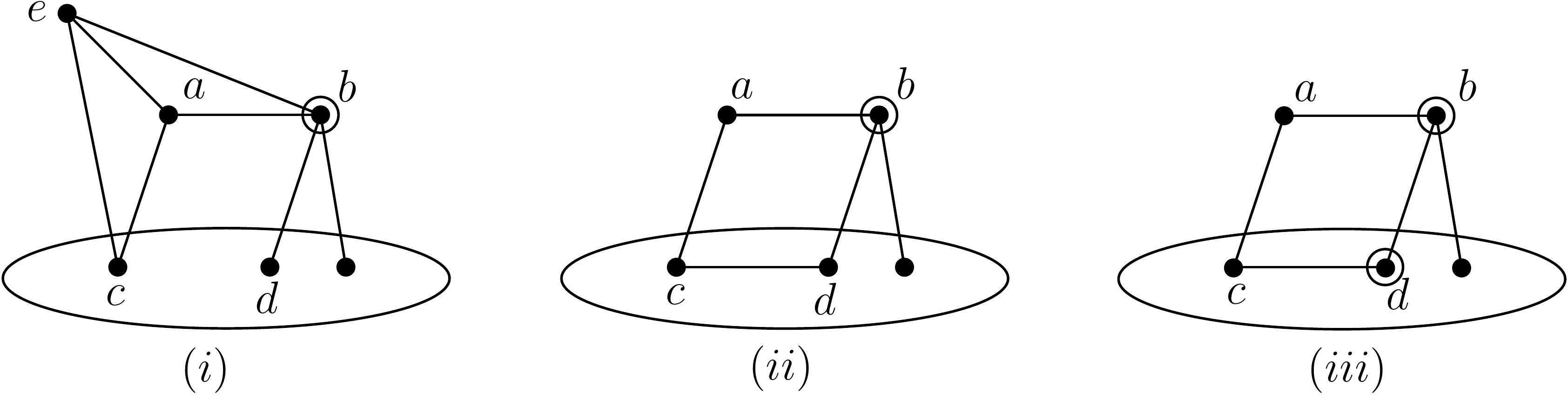} 
\caption{Cases when an edge contains an unmarked vertex} 
\label{fig:edge-cases}
\end{figure}

\begin{lemma}
\label{lemma:edge-cases}
If there exists a component in $G-S$ that comprises of an edge $(a,b)$, and $b$ belongs to an optimum tracking set for $G$, then $a$ need not be marked as a tracker.
\end{lemma}

\begin{proof}
Let $c$ be the neighbour of $a$ in $S$. We claim that if $a$ is marked as a tracker, it can be replaced by $c$ as a tracker. If no optimum tracking set contains $a$, then the claim holds. Suppose not. Observe that both $c$ and $b$ are trackers. Then $a$ needs to belong to a tracking set only if there exists one more untracked path, say $P_u$,  between $c$ and $b$. Path $P_u$ can not have vertices only from $G-S$, as this contradicts the assumption that $(a,b)$ is a single edge component in $G-S$. See Figure~\ref{fig:edge-cases}$(i)$. Thus $P_u$ contains vertices from $G$. See Figure~\ref{fig:edge-cases}$(ii)$. We can mark a vertex from $S$ to distinguish $P_u$. Suppose all vertices in $V(P_u)\cap S$ are already marked as trackers. See Figure~\ref{fig:edge-cases}$(iii)$. Observe that in such a case, $P_u$ is already differentiated. Further, since $a$ is adjacent to only $c$ and $b$, $b,c$ can help distinguish the paths passing through $a$ by the sequence of their vertices. Hence, an optimum tracking set for $G$ need not contain $a$.
\qed
\end{proof}

Now all vertices in $G-S$ are either already marked as trackers, or need not be marked as trackers, or have been added to the set $X$. 
We consider all vertices in $X$, and apply the algorithm for DCM. In the final step of the DCM algorithm, where we consider all subsets of unmarked vertices and verify if they form a tracking set (along with the marked ones), we include the vertices in $Y$ along with the unmarked vertices. This does not affect the bounds as $|Y|\leq 2$. Observe that we ignore all other vertices in $G-S$, except the ones in $X$, while analyzing different cases in the DCM algorithm. Hence, we have the following theorem.

\begin{theorem}
\label{theorem:tp-cvd}
For a graph $G$ with a known cluster vertex deletion set of size $k$, \tp can be solved in $2^{\Oh(k^2)}.n^{\Oh(1)}$ time.
\end{theorem}

\section{Conclusions}

In this paper, we study structural parameterizations of the \tp problem. We prove that \tp is FPT when parameterized by the size of vertex cover or the size of cluster vertex deletion set. We do so by giving a generalized algorithm for the case when the modulator has some specific properties.
It would be interesting to explore if the running time of our algorithms can be improved further.
Future scope involves studying \tp structural parameterization with respect to more parameters like odd cycle traversal, feedback vertex set and distance to chordal graph.

\bibliographystyle{splncs04}
\bibliography{tracking}

\end{document}